%% file: ms.tex
\documentclass[10pt,conference]{IEEEtran}

\usepackage{cite}

\usepackage{amsmath,amsthm,amsfonts,amssymb,amsbsy,esint,cancel}

\newtheorem{lemma}{Lemma}
\newtheorem{prop}{Proposition}

\usepackage{graphicx}
\usepackage{xcolor}
\def\BibTeX{{\rm B\kern-.05em{\sc i\kern-.025em b}\kern-.08em
    T\kern-.1667em\lower.7ex\hbox{E}\kern-.125emX}}

\usepackage{textcomp,gensymb,bbm,bm,dsfont,upgreek}

\usepackage[caption=false,font=footnotesize]{subfig}

\usepackage{cleveref}
\crefname{figure}{Fig.}{Fig.}
\Crefname{figure}{Fig.}{Fig.}
\crefname{section}{Sec.}{Sec.}
\Crefname{section}{Sec.}{Sec.}
\crefname{subsection}{Sec.}{Sec.}
\Crefname{subsection}{Sec.}{Sec.}
\crefname{theorem}{Theorem}{Theorem}
\Crefname{theorem}{Theorem}{Theorem}
\crefname{lemma}{Lemma}{Lemma}
\Crefname{lemma}{Lemma}{Lemma}
\crefname{prop}{Prop.}{Prop.}
\Crefname{prop}{Prop.}{Prop.}
\crefname{equation}{}{}
\Crefname{equation}{}{}

\usepackage{balance} 

\newcommand\DoSixPageSubmission

\input{macros}
\begin{document}

\title{The Channel Between Randomly Oriented Dipoles: Statistics and Outage in the Near and Far Field}
\author{
\IEEEauthorblockN{Gregor Dumphart and Armin Wittneben} 
\IEEEauthorblockA{\textit{Wireless Communications Group, D-ITET, ETH Zurich}\\
Z\"urich, Switzerland\\
Email: \{dumphart, wittneben\}@nari.ee.ethz.ch}}



\maketitle

\begin{abstract}
We consider the class of wireless links whose propagation characteristics are described by a dipole model. This comprises free-space links between dipole antennas and magneto-inductive links between coils, with important communication and power transfer applications. A dipole model describes the channel coefficient as a function of link distance and antenna orientations. In many use cases the orientations are random, causing a random fading channel. This paper presents a closed-form description of the channel statistics and the resulting outage performance for the case of i.i.d. uniformly distributed antenna orientations in 3D space. For reception in AWGN after active transmission, we show that the high-SNR outage probability scales like $p_\mathrm{e} \propto \mathrm{SNR}^{-1/2}$ in the near- or far-field region, i.e. the diversity exponent is just 1/2 (even 1/4 with backscatter or load modulation). The diversity exponent improves to 1 in the near-far-field transition due to polarization diversity. Analogous statements are made for the power transfer efficiency and outage capacity. 
\end{abstract}

\begin{IEEEkeywords}
fading, outage, polarization diversity, dipole antennas, loop antennas, coil misalignment, backscatter
\end{IEEEkeywords}

\section{Introduction}
\input{01-Intro}

\section{Dipole Channel Model}\label{sec:formalism}
\input{02-Model}

\section{Channel Statistics}\label{sec:stats}
\input{03-Stats}

\section{Outage Analysis}\label{sec:outage}
\input{04-Outage}

\ifdefined\DoSixPageSubmission
\else

\section{Further Discussion}

\subsection{Comparison to Rayleigh Fading}
\label{sec:rayleigh}

Rayleigh fading $h \sim \mathcal{CN}(0,\sigma^2)$ is an established model for multipath radio channels with rich scattering \cite{Tse2005}.
It exhibits $F_{|h|^2}(s) = 1- e^{-s/\sigma^2}$ and thus $F_{|h|^2}(s) \approx \f{s}{\sigma^2}$ for small $s$.
The PDF support is the entire $\mathbb{C}$-plane. Likewise, the PDF supports observed in
\Cref{fig:hTransConditional,fig:hFullyRandOrScatter_b,fig:hFullyRandOrScatter_c,fig:FinalPdfNumInt} 
are two-dimensional manifolds, centered at $h=0$. The two dimensions are due to \textit{phase-shifted field components with linear independence} at the RX position, which \textit{yield polarization diversity}. For the links studied in this paper, these field components arise from the non-coherent superposition of near- and far-field propagation. For Rayleigh fading, they arise from the non-coherent superposition of different paths. Hence, it is intuitive that either case features a diversity exponent of $1$.

The near- and far-field regions do not bring polarization diversity: the PDF support degenerates to a one-dimensional manifold in the $\mathbb{C}$-plane and the diversity exponent drops to $1/2$ because realizations $h \approx 0$ become more likely. The asymptotic outage behavior is analogous to that of a channel coefficient with a real-valued Gaussian distribution, i.e. just the real part of a Rayleigh-distributed random variable.


\fi

\ifdefined\DoSixPageSubmission
\section{Implications for RFID and Backscatter}
\else
\subsection{Implications for RFID and Backscatter}
\fi

\label{sec:rfid}

So far, the results concerned links with an active TX equipped with a TX amplifier, where $\SNR \approx \SNR\Opt J_*^2$ with $J_* \in \{J\nf, J\ff\}$ in the near- or far-field region. For a passive RFID tag that uses load modulation or for backscatter communication, the fading channel applies twice and the relation changes to
$\SNR \approx \SNR\Opt J_*^4$, cf. \cite{DumphartPIMRC2016Short}, with the following severe consequences. The misalignment losses double in terms of $\mathrm{dB}$ value (e.g., the abscissa of \Cref{fig:StatMotivation}). Likewise, the bit error rate $\BER \propto \SNR\Opt^{-1/4}$ with a diversity exponent of only $1/4$.
\ifdefined\DoSixPageSubmission
\else
In the context of NFC payments, this can however be seen as a security feature: severe attenuate is likely for unintended links.
\fi

\ifdefined\DoSixPageSubmission
\else

\subsection{Improvements from Orthogonal Arrays}
\label{sec:arrays}

Consider a $3 \times 3$ MIMO link between two colocated arrays, each consisting of three orthogonal dipoles. In \cite{DumphartPIMRC2016Short,Dumphart2020} we showed that the performance of such a link is unaffected by the array orientations for the case that optimal beamforming (maximum-ratio combining) is used at both ends. The resultant $\SNR$ equals $\SNR\Opt$, like for an optimally arranged SISO link. In other words, fading is mitigated entirely.

Now consider a SIMO link with a colocated RX array of three orthogonal dipoles, using maximum-ratio combining. The system performs like a SISO link with $|h| = \|\v\|$ (see \cite{Dumphart2020}). When the field vector $\v$ from \Cref{eq:v} has an appreciable near-field component then fading is mitigated entirely, because of the lower-bounded magnitude $\f{1}{2} \leq \beta\nf$. However, fading does occur in the far-field region because $\|\v\| \propto \beta\ff$ does fade when $\o\Tx = \pm\DoD$. This fading channel exhibits a diversity exponent of $1$ because $f(\beta\ff) \approx \beta\ff$ for small $\beta\ff$, cf. \Cref{eq:magnitudeDistrFf}.

\fi

\section{Summary \& Conclusions}
\label{sec:summary}

We provided an analytical description of the statistics of the fading channel between two randomly oriented dipoles. Our outage analysis revealed that drastic signal losses are very likely to occur, especially in the near- and far-field regions. This emphasizes the importance of diversity concepts for this class of links, e.g., the use of a rotating source, appropriate antenna polarization, or antenna arrays and beamforming.
The results also suggest that, for special applications, it can be favorable to design a link such that the RX will typically be located in the transition region between near and far field.

\appendices

\section*{Appendix: Physical Conditions \& Details}
\input{99-Appendix}

\section*{Acknowledgement}
We would like to thank Robin Kramer for valuable inputs and Bharat Bhatia for helping with the appendix.

\bibliographystyle{IEEEtran}
\bibliography{../../GD}

\end{document}

%% file: macros.tex
\newcommand\tn[1]{\text{\normalfont{#1}}} 

  \newcommand\f[2]{\frac{#1}{#2}} 
  \newcommand\re{\mathrm{Re}} 
  \newcommand\im{\mathrm{Im}} 
\renewcommand\Re{\re} 
\renewcommand\Im{\im} 

\DeclareMathOperator\arcosh{arcosh}

  \newcommand\IndicatorFunc{\mathds{1}}

  \newcommand\mtx[2]{\left[\begin{array}{#1}#2\end{array}\right]}
  \newcommand\eye{\mathbf{I}}     
  \newcommand\Tr{^\mathrm{\normalfont{{T}}}}     

\newcommand\SymbolOfProb{\mathrm{P}} 
\newcommand\Prob[1]{\SymbolOfProb[\hspace{.35mm}#1\hspace{.35mm}]}
\newcommand\ProbLR[1]{\SymbolOfProb\left[\hspace{.35mm}#1\hspace{.35mm}\right]}
\newcommand\SymbolOfEV{\mathbb{E}} 
\newcommand\EVVar[2]{\SymbolOfEV_{#1}[\hspace{.2mm}#2\hspace{.2mm}]}
\newcommand\EV[1]{\EVVar{}{#1}}

\newcommand\iid{\overset{\tn{i.i.d.}}{\sim}}
\newcommand\UD{\mathcal{U}} 

\renewcommand\o{\mathbf{o}}   
\newcommand\DoD{\mathbf{u}}

\newcommand\unit[1]{\,\mathrm{#1}}	
\newcommand\dB{\unit{dB}}

\renewcommand\b{\boldsymbol\upbeta} 
\newcommand\nfs[1]{_{\tn{NF}#1}} 
\newcommand\ffs[1]{_{\tn{FF}#1}} 
\newcommand\nf{\nfs{}}
\newcommand\ff{\ffs{}}
\newcommand\ffSPACE{\ff\hspace{.27mm}} 

\newcommand\coeffH{\alpha} 

\newcommand\SymbolOfTx{\hspace{.2mm}\tn{T}}
\newcommand\SymbolOfRx{\hspace{.2mm}\tn{R}}
\newcommand\Txs[1]{_{\SymbolOfTx#1}} 
\newcommand\Rxs[1]{_{\SymbolOfRx#1}} 

\newcommand\Tx{\Txs{}} 
\newcommand\Rx{\Rxs{}} 

\newcommand\SNR{\mathrm{SNR}}

\renewcommand\v{\mathbf{v}} 
\newcommand{\subRe}{_\tn{Re}}
\newcommand{\subIm}{_\tn{Im}}

  \newcommand\SymbolOfN{\mathrm{N}}

  \newcommand\Ns[1]{_{\SymbolOfN#1}}

  \newcommand\N{\Ns{}}

\newcommand\MakeNumberRoman[1]{\uppercase\expandafter{\romannumeral#1}}

\newcommand\myCustomVec{}

\newcommand{\CircFunc}{\psi} 
\newcommand{\bEquiv}{\bar\beta\nf}
\newcommand\Coax{_\tn{coax}}
\newcommand\Copl{_\tn{para}}
\newcommand\CoplSPACE{_{\tn{para}\hspace{.43mm}}}
\newcommand\Opt{_\tn{opt}}

\newcommand\NTurn{N}
\newcommand\NTurnTx{\NTurn\Tx}
\newcommand\NTurnRx{\NTurn\Rx}

\newcommand\OutProb{\epsilon}
\newcommand\BER{p_\tn{e}}

%% file: 01-Intro.tex
Wireless engineers often rely on statistical channel models to describe complicated propagation environments and their fading characteristics. Fading occurs when (for a narrowband channel) the random channel coefficient $h \in \mathbb{C}$ is close to zero, i.e. $h \approx 0$, which can cause a link outage \cite{Tse2005}.
A statistical channel model allows to study the performance implications of fading analytically.
For example, for a digital modulation scheme with transmit power $P$, Rayleigh fading $h \sim \mathcal{CN}(0,\sigma^2)$ due to rich multipath propagation, and additive white Gaussian noise (AWGN), the bit error rate is asymptotically proportional to $P^{-L}$ with $L = 1$. The small diversity exponent $L$ means that Rayleigh fading hinders reliable low-power communication \cite[Cpt.~3]{Tse2005}. The existing literature states many such performance results for various models, e.g. see \cite{BiglieriTIT1998}.

Fading events $h \approx 0$ can be caused by mechanisms other than multipath propagation or shadowing; they can even occur in free space in the case of inopportune antenna orientations. For example, the transmitter-to-receiver (TX-to-RX) direction may coincide with a zero of the TX-antenna radiation pattern or the RX antenna may be misaligned with the incident field \cite{Balanis2005}. To that effect, non-isotropic antennas with random orientations can give rise to a fading channel with random channel coefficient $h$, even in free space \cite{CoxTC1983,DietrichTAP2001}.
Such random antenna orientations  are to be expected for wireless applications with high mobility or application-specific node locations.
Associated fading channels have been studied in \cite{CoxTC1983} for mobile radio devices and in \cite{DumphartPIMRC2016Short,ZhangTWC2017} for magnetic induction links between randomly arranged coils.

In this paper we study fading due to random antenna orientations for the class of wireless links that can be adequately modeled by a free-space link from a TX dipole to a RX dipole. This comprises:
\begin{itemize}
\item Magnetic induction links between two weakly coupled coils (loop antennas).
\item Capacitive links between small electric dipole antennas.
\item Links between $\lambda/2$-length dipole antennas in free space.
\item Links from a magnetic dipole to a small loop or from an electric dipole to a pair of terminals with small separation.
\end{itemize}
Those have important applications in wireless power transfer and data communication, either with an active TX or a passive tag (RFID load modulation or backscatter modulation). \cite{Finkenzeller2015}

The statistics and communication-theoretic performance aspects of this fading channel are, to the best of our knowledge, not covered by existing work.
The need for an appropriate statistical channel model was highlighted in \cite{AngererRuppTC2010} where, because of the lack of a better model, a Rayleigh fading model was assumed for RFID links. Similarly, the heuristic assumption of a Gaussian-distributed data rate was made in \cite{SunTC2013} for a randomly arranged magnetic induction link.

\subsubsection*{Contribution}
This paper contains the following novel results for links between dipoles with uniformly distributed 3D orientations, presented in communication-theoretic parlance.
\begin{itemize}
\item We derive the channel statistics for the near- and far-field region. We show that the outage behavior is characterized by a diversity exponent of just $\f{1}{2}$.
\item We derive the channel statistics in the near-far-field transition and demonstrate a diversity exponent of $1$.
\item An outage analysis demonstrates the severity of this fading channel in terms of the behavior of the outage PTE, outage capacity, and bit error probability.
\end{itemize}

\subsubsection*{Related Work}
Regarding misaligned magnetic-induction links, most studies focused on small lateral or angular deviations in the regime of short-range power transfer \cite{Soma1987,Fotopoulou2011} where the specific coil geometries must be considered.
The concept of outage probability, diversity exponent and outage capacity in relation to fading is well-established for multipath radio channels \cite{Tse2005}. Likewise, polarization diversity is a well-established concept \cite[Sec.~2.5]{Orfanidis2002}.
The work in \cite{SunTC2013} identified the outage capacity as a meaningful performance measure of randomly arranged magneto-inductive communication links.
The distribution in \Cref{eq:PdfJNf} and various results on diversity combining appeared in our paper \cite{DumphartPIMRC2016Short}.
The contents of this paper are also contained in the dissertation of the first author \cite{Dumphart2020}.

\subsubsection*{Paper Structure}
\Cref{sec:formalism} describes the dipole model in detail. The rather technical \Cref{sec:stats} then derives the channel statistics between randomly oriented antennas, which enables the subsequent outage and diversity analysis in \Cref{sec:outage}. After 
commenting on the implications for RFID and backscatter systems in \Cref{sec:rfid}, we conclude the paper in \Cref{sec:summary}.

%% file: 02-Model.tex
We consider a narrowband wireless link from a transmitting dipole, driven by a TX amplifier, to a receiving dipole which feeds a low-noise amplifier or tank circuit. The link geometry is shown in \Cref{fig:RandomOrConcept} and is described by the link distance $r$ and three unit vectors: the TX and RX dipole axis directions $\o\Tx, \o\Rx \in \mathbb{R}^3$ and the TX-to-RX direction $\DoD \in \mathbb{R}^{\scriptscriptstyle 3}$.

\begin{figure}[!ht]\centering
\includegraphics[width=.70\columnwidth,trim=0 0 0 0]{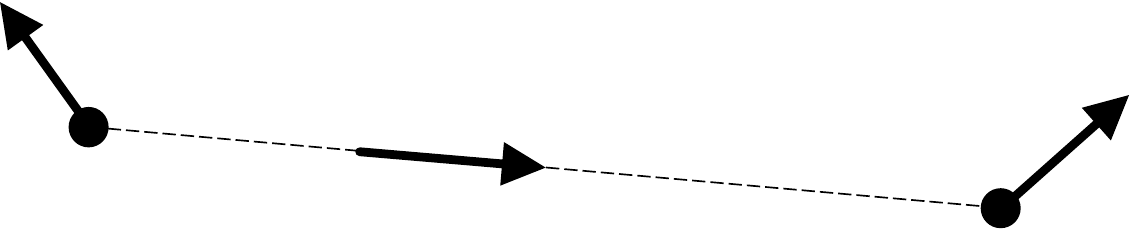}
\put(-180,48){$\substack{\text{transmit dipole}\\[.5mm]\text{orientation}}$}
\put(-165,33){$\o\Tx$}
\put(-34,41){$\substack{\text{receive dipole}\\[.5mm]\text{orientation}}$}
\put(-19,26){$\o\Rx$}
\put(-120,31){$\substack{\text{TX-to-RX}\\[.9mm]\text{direction}}$}
\put(-110,16){$\DoD$}
\put(-69,25){$\substack{\text{link}\\[.9mm]\text{distance}}$}
\put(-60,10){$r$}
\caption{Geometry of a link between dipoles with arbitrary orientations $\o\Tx, \o\Rx \in \mathbb{R}^3$ (unit vectors) and distance $r$. The model  serves as a description of unaligned links between weakly coupled coils or between dipole antennas.}
\label{fig:RandomOrConcept}
\end{figure}
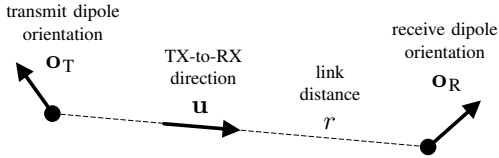

The channel coefficient $h \in \mathbb{C}$ is given by \cite{Balanis2005,Dumphart2020}
\begin{align}
h
= \coeffH \left(\!\left( \f{1}{(kr)^3} + \f{j}{(kr)^2} \right) J\nf 
+ \f{1}{2kr}\, J\ff \right)
\label{eq:hDipole}
\end{align}
where $k = \f{2\pi}{\lambda} = \f{2\pi f}{c}$ is the wavenumber and $j$ the imaginary unit.
The prefactor $\coeffH$ is of no formal importance for this paper. It is given by $\coeffH = \bar\coeffH \,e^{-jkr}$ where $\bar\coeffH \in \mathbb{C}$ subsumes technical parameters (e.g., coil diameters) which are described in the appendix together with the detailed model conditions.

The RX may be located in the near-field region ($kr \ll 1$) or the far-field region ($kr \gg 1$), or in the transition in between.
The formula \Cref{eq:hDipole} uses the near- and far-field alignment factors,
given by the inner products
\begin{align}
&J\nf = \o\Rx\Tr \b\nf 
\, , &
&J\nf \in [-1,1]
\, , \label{eq:JNfDef} \\
&J\ffSPACE = \o\Rx\Tr \b\ffSPACE 
\, , &
&J\ffSPACE \in [-1,1]
\, . \label{eq:JFfDef}
\end{align}
They account for signal attenuation due to suboptimal node orientations (misalignment). The formulas use unitless field vector quantities $\b\nf$ and $\b\ff$, which we call the scaled near field
and the scaled far field, respectively. They are given by
\begin{align}
&\b\nf = \f{1}{2}( 3\DoD\DoD\Tr - \eye_3 )\,\o\Tx
\, , \label{eq:bNfDef} \\
&\b\ffSPACE = (\eye_3 - \DoD\DoD\Tr)\o\Tx
\label{eq:bFfDef}
\end{align}
and illustrated in \Cref{fig:beta}.
The formulas use a convenient linear-algebraic formalism, which has been derived in our previous work \cite[App.~A]{Dumphart2020} from an existing trigonometric description of the dipole field \cite{Balanis2005}.
The field magnitudes are given by
\begin{align}
&\beta\nf = \|\b\nf\| = \f{1}{2} \sqrt{1\!+\!3(\DoD\Tr \o\Tx)^2}
\, , && 
\tfrac{1}{2} \leq \beta\nf \leq 1
\, , \label{eq:bNfMagn} \\
&\beta\ffSPACE = \|\b\ff\| = \sqrt{1\!-\!(\DoD\Tr \o\Tx)^2}
\, , &&
0 \leq \beta\ff \leq 1 \, .
\label{eq:bFfMagn}
\end{align}
The far-field magnitude $\beta\ff$ can fade to zero but $\beta\nf$ can not, as can be seen in \Cref{fig:beta}.

\begin{figure}[!h]\centering
\subfloat[scaled near field $\b\nf$]{
\includegraphics[width=.44\columnwidth]{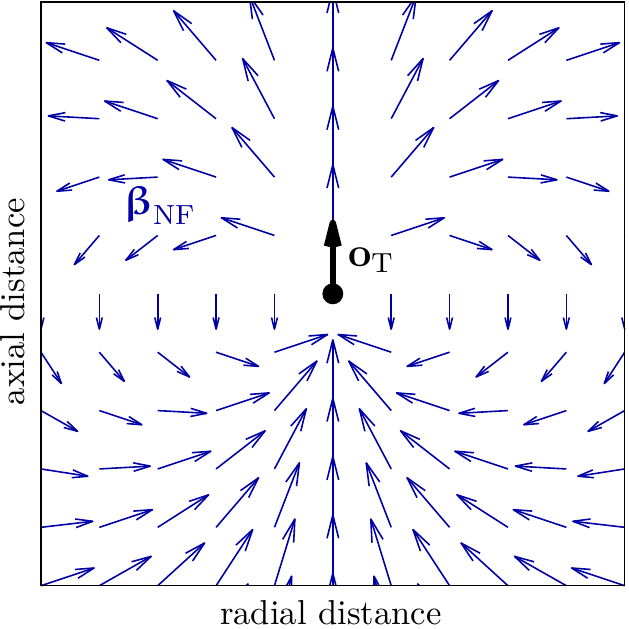}
}\ \ \,
\subfloat[scaled far field $\b\ff$]{
\includegraphics[width=.44\columnwidth]{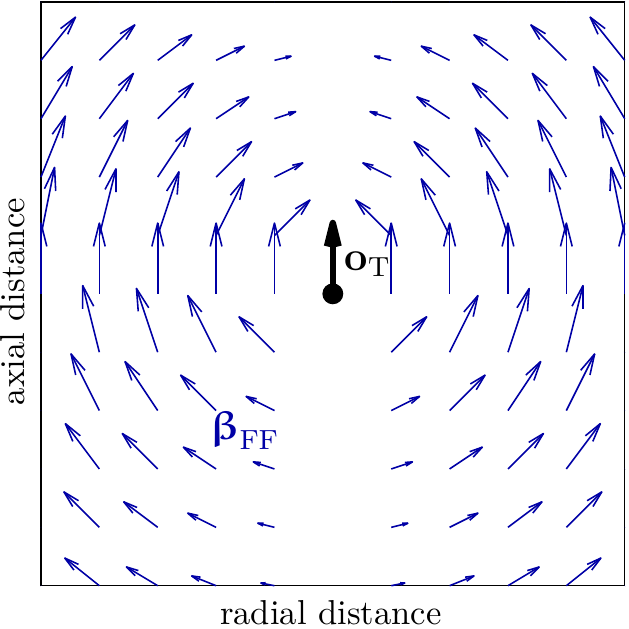}
}
\caption{Scaled near and far field around a transmitting dipole with vertical axis orientation (unit vector $\o\Tx$). By definition these fields do not comprise path loss; the maximum magnitude is $1$. In particular, $\beta\nf = 1, \beta\ff = 0$ hold on the dipole axis while $\beta\nf = \f{1}{2}, \beta\ff = 1$ hold in the perpendicular plane.}
\label{fig:beta}
\end{figure}

We shall point out two specific dipole arrangements:
\begin{itemize}
\item Dipoles in coaxial arrangement $\o\Tx = \o\Rx  = \DoD$:\\
in this case $J\nf = 1$, $J\ff = 0$, and thus $h = h\Coax$ with
\begin{align}
h\Coax =
\coeffH \left( \f{1}{(kr)^3} + \f{j}{(kr)^2} \right) .
\label{eq:hCoax}
\end{align}
\item Dipoles in parallel arrangement $\o\Tx = \o\Rx$ with $\DoD\Tr \o\Tx = 0$:
in this case $J\nf = -\tfrac{1}{2}$, $J\ff = 1$ and thus $h = h\Copl$ with
\begin{align}
h\Copl =
\f{\coeffH}{2} \left( -\f{1}{(kr)^3} - \f{j}{(kr)^2} + \f{1}{kr}  \right) .
\label{eq:hCopl}
\end{align}
\end{itemize}

We note that $\eta = |h|^2$ is the power transfer efficiency (PTE) over the link.
An important quantity is the maximum PTE given $kr$ and $\coeffH$, denoted as
$\eta\Opt = |h\Opt|^2$.
We find that\footnotemark{}
\begin{align}
\eta\Opt
=
|h\Opt|^2
=
\max_{\o\Tx,\o\Rx} |h|^2
=
\left\{\begin{array}{ll}
|h\Coax|^2 & \tn{if}\ kr \leq kr_\tn{th} \\
|h\CoplSPACE|^2 & \tn{if}\ kr  >   kr_\tn{th}
\end{array}\right.
\label{eq:hOpt}
\end{align}
whereby the threshold fulfills $|h\Coax| = |h\Copl|$. It is given by
\begin{align}
kr_\tn{th} = \sqrt{\f{\sqrt{37}+5}{2}} \approx 2.3540 \, .
\label{eq:krThresh}
\end{align}

\footnotetext{To prove the statement \Cref{eq:hOpt}, we write $h$ as bilinear form $h = \o\Rx\Tr {\bf A} \o\Tx$ and deduce
${\bf A} = \coeffH ( ( \f{1}{(kr)^3} + \f{j}{(kr)^2} ) ( \f{3}{2}\DoD\DoD\Tr - \f{1}{2}\eye_3 )
+ \f{1}{2kr}(\eye_3 - \DoD\DoD\Tr) ) \in \mathbb{C}^{3 \times 3}$
from \Cref{eq:hDipole,eq:JNfDef,eq:JFfDef,eq:bNfDef,eq:bFfDef}.
We find that $h\Coax$ is an eigenvalue by verifying ${\bf A}\DoD = h\Coax \DoD$. 
Furthermore, $h\Copl$ is a double eigenvalue because ${\bf A}\DoD_\perp = h\Copl \DoD_\perp$ for any vector $\DoD_\perp$ that is orthogonal to $\DoD$.
Therefrom, the statement \Cref{eq:hOpt} follows from basic linear algebra.
The threshold $kr_\tn{th}$ in \Cref{eq:krThresh} is found by solving the equation $|h\Coax|^2 = |h\Copl|^2$ for $kr$, using the definitions \Cref{eq:hCoax,eq:hCopl}.}

%% file: 03-Stats.tex
Our starting point is the assumption that the TX and RX antenna orientations (unit vectors) are random and statistically independent, with uniform distributions
\begin{align}
\o\Tx, \o\Rx \iid \, \UD(\mathcal{S})
\label{eq:UniformOrient}
\end{align}
on the unit sphere $\mathcal{S} \subset \mathbb{R}^3$. The quantities $\coeffH,k,r,\DoD$ are considered non-random throughout. Hence, the statistics of $h$ in  \Cref{eq:hDipole} are determined by the joint statistics of $J\nf$, $J\ff$.

\subsection{In the Near-Field Region or Far-Field Region}

First, we address the important marginal distributions of $J\nf$ and $J\ff$ which describe the statistics of $h$ in the near-field region ($kr \ll kr_\tn{th}$) and the far-field region ($kr \gg kr_\tn{th}$), respectively.

\begin{prop}\label{prop:JDistr}
Assume \Cref{eq:UniformOrient}.
Then the near-field alignment factor $J\nf$ has the marginal probability density function (PDF)
\begin{align}
f_{J\nf}(J\nf) = \f{1}{2\bEquiv} \cdot \left\{\!\begin{array}{lrl}
1                                    & &|J\nf| \leq \f{1}{2} \\
1 - \f{\arcosh\left(2|J\nf|\right)}{\arcosh(2)} \!\! & \f{1}{2} < \!\!\!\! &|J\nf| <  1        \\
0                                    & 1 \leq \!\!\!\! &|J\nf|
\end{array}\right.
\label{eq:PdfJNf}
\end{align}
with $\bEquiv = \f{\sqrt{3}}{2\,\mathrm{arcosh}(2)}$.
The far-field alignment factor exhibits
\begin{align}
f_{J\ff}(J\ff) &= 
\frac{1}{2} \left( \frac{\pi}{2} - \arcsin|J\ff| \right) \cdot \IndicatorFunc_{[-1,1]}(J\ff) \, .
\label{eq:PdfJFf}
\end{align}
\end{prop}

{\noindent}Thereby, $\IndicatorFunc_{[-1,1]}$ is the indicator function for this interval. The PDFs are shown in \Cref{fig:jnfpdf,fig:jffpdf}.

\begin{figure}[!ht]
\vspace{-1mm}
\centering
\subfloat[PDF of alignment factor $J\nf$]{\label{fig:jnfpdf}
\includegraphics[width=.45\columnwidth]{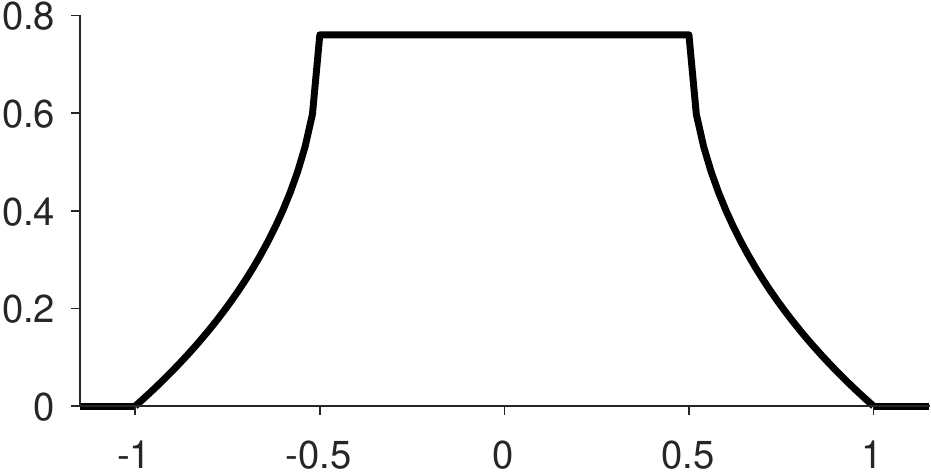}
\put(-99,47){$f_{J\nf}$}}\
\subfloat[PDF of alignment factor $J\ff$]{\label{fig:jffpdf}
\includegraphics[width=.45\columnwidth]{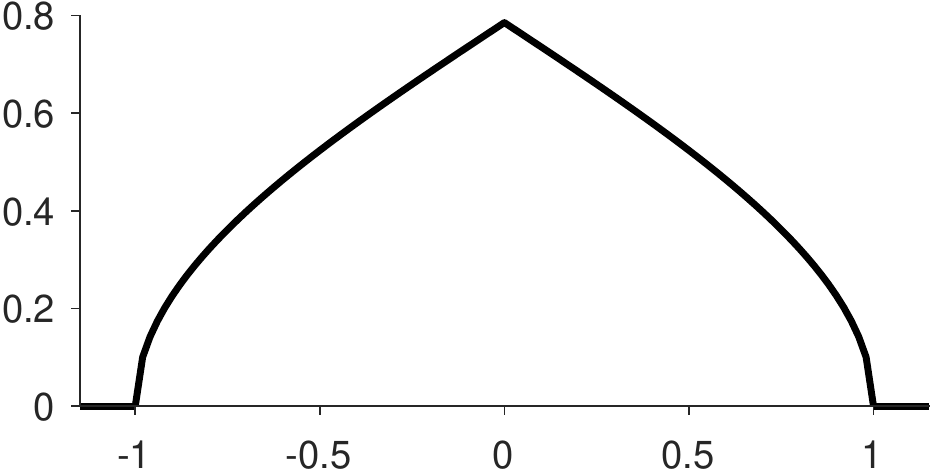}
\put(-93,47){$f_{J\ff}$}} \\
\subfloat[PDF of magnitude $\beta\nf$]{\label{fig:betanfpdf}
\includegraphics[width=.44\columnwidth]{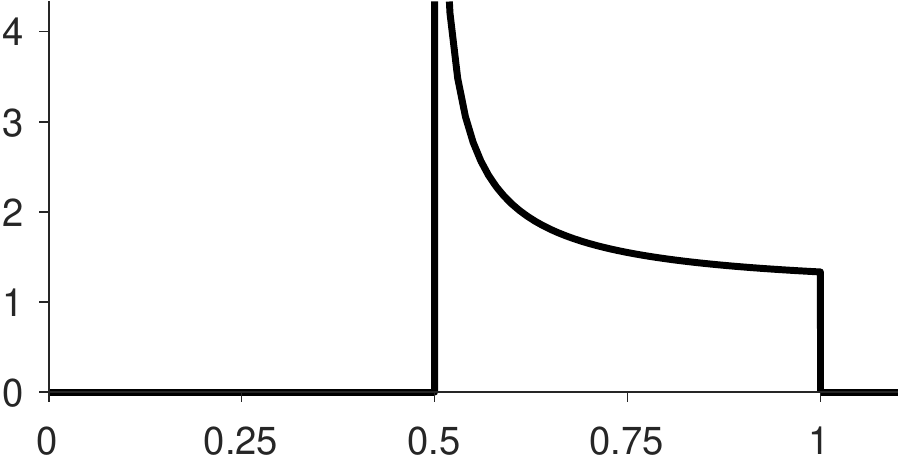}
\put(-43,40){$f_{\beta\nf}$}}\ \
\subfloat[PDF of magnitude $\beta\ff$]{\label{fig:betaffpdf}
\includegraphics[width=.44\columnwidth]{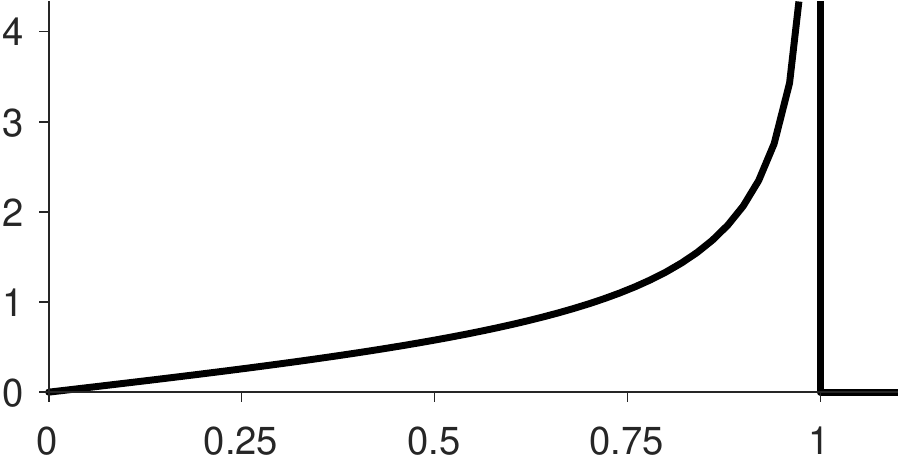}
\put(-41,40){$f_{\beta\ff}$}} 
\caption{Marginal PDFs arising from random antenna orientations on both ends with uniform distributions in 3D.}
\label{fig:jpdfMajor}
\end{figure}

\begin{proof}
We will heavily use the fact that, for a random constant-length vector in $\mathbb{R}^3$ with uniform distribution on a sphere, any projection has uniform distribution. This fact is a corollary of Archimedes' hat-box theorem or of the fact that the lateral surface area of a sphere cap is linear in its height (which implies a linear CDF for a projection, cf. \cite{DumphartPIMRC2016Short}).
A first implication to our formalism is that the TX-side projection $\DoD\Tr \o\Tx$, which determines the magnitudes $\beta\nf$ and $\beta\ff$, has uniform distribution $\DoD\Tr \o\Tx \sim \UD(-1,1)$ due to $\o\Tx \sim \UD(\mathcal{S})$. Consequently, with a basic change-of-variables calculation we obtain from \Cref{eq:bNfMagn,eq:bFfMagn} the PDFs of the field magnitudes
\begin{align}
&f_{\beta\nf}(\beta\nf) 
= \f{4}{\sqrt{3}} \f{\beta\nf}{\sqrt{4\beta\nf^2 - 1}}
\cdot \IndicatorFunc_{[\f{1}{2},1]}(\beta\nf)
\, , \label{eq:magnitudeDistrNf} \\
&f_{\beta\ff}(\beta\ff) 
= \f{\beta\ff}{\sqrt{1 - \beta\ff^2}}
\cdot \IndicatorFunc_{[0,1]}(\beta\ff)
\label{eq:magnitudeDistrFf}
\end{align}
which are shown in \Cref{fig:betanfpdf,fig:betaffpdf}. The random RX orientation $\o\Rx \sim \UD(\mathcal{S})$ in \Cref{eq:JNfDef,eq:JFfDef} results in conditional distributions 
$J\nf \, | \, \beta\nf \sim \UD(-\beta\nf,\beta\nf)$
and
$J\ff \, | \, \beta\ff \sim \UD(-\beta\ff,\beta\ff)$. The joint PDFs
$f_{J\nf | \beta\nf} \cdot f_{\beta\nf}$ and $f_{J\ff | \beta\ff} \cdot f_{\beta\ff}$
yield $f_{J\nf}$ and $f_{J\ff}$ via marginalization integrals (the steps are omitted).
\end{proof}

\begin{prop}
\label{prop:MarginalRegionCDF}
Consider the cumulative distribution function (CDF) $F_{|h|^2}(s) = \Prob{|h|^2 \leq s}$. In the near-field region $kr \ll kr_\tn{th}$ (described by $J_* = J\nf$) or the far-field region $kr \gg kr_\tn{th}$ (described by $J_* = J\ff$), the approximation
\begin{align}
F_{|h|^2}(s)
\approx
\f{2 \cdot f_{J_*}\!(0)}{|h\Opt|}\, \sqrt{s}
\label{eq:MarginalRegionCDF}
\end{align}
applies under assumption \Cref{eq:UniformOrient}. It is accurate for $s \ll |h\Opt|^2$.
\end{prop}

\begin{proof}
Either case fulfills $|h|^2 \approx |h\Opt|^2 J_*^2$. We calculate
\begin{align}
F_{|h|^2}(s) &\approx F_{J_*^2}\Big( \f{s}{|h\Opt|^2} \Big)
= \ProbLR{ |J_*| \leq\! \f{\sqrt{s}}{|h\Opt|} }
\nonumber \\
&= 2 \int_0^{\f{\sqrt{s}}{|h\Opt|}} f_{J_*}\!(x) \,dx
\leq 2\cdot f_{J_*}\!(0) \int_0^{\f{\sqrt{s}}{|h\Opt|}} dx \, .
\end{align}
This bound is tight for small integration intervals because the integrand is continuous, as seen in \Cref{fig:jnfpdf,fig:jffpdf}.
\end{proof}

The CDF behavior $F_{|h|^2}(s) \propto \sqrt{s}$ for small $s$ hints that fading events $|h|^2 \approx 0$ occur with significant probability. This is due to the probability densities $f_{J\nf}(0) > 0$, $f_{J\ff}(0) > 0$. 


\subsection{Near-Far Transition with Random Receiver Orientation}

We consider the statistics of $h \in \mathbb{C}$ when both near- and far-field propagation make significant contributions. First, we consider the case of a random RX orientation $\o\Rx \sim \UD(\mathcal{S})$ while $\o\Tx$ and $\DoD$ are fixed. This interesting case will serve as preparation for the fully random case.

We start our mathematical approach by observing from \Cref{eq:hDipole,eq:JNfDef,eq:JFfDef} that the channel coefficient is an inner product 
\begin{align}
& h       = \o\Rx\Tr \v \, , &&
\o\Rx \in \mathbb{R}^3 , \
\v    \in \mathbb{C}^3
\label{eq:ChanCoeffProj}
\end{align}
of the random $\o\Rx$ and a unitless, complex-valued field vector
\begin{align}
\v = \coeffH \left(\!\left(\f{1}{(kr)^3} \!+\! \f{j}{(kr)^2} \!\right) \!\b\nf  \!+\! \f{1}{2kr} \b\ff \!\right) .
\label{eq:v}
\end{align}
This field vector is non-random in this context because it is determined by the non-random $\coeffH, kr, \DoD, \o\Tx$. We consider $\v\subRe = \Re(\v)$ and $\v\subIm = \Im(\v)$ and note that these two vectors are linearly independent unless \mbox{$\DoD\Tr \o\Tx = 0$} or $\DoD\Tr \o\Tx = \pm 1$; the simple proof thereof is omitted. 

The random channel coefficient is expressed as
\begin{align}
h =  \Re(h) + j \cdot \Im(h) = \o\Rx\Tr \v\subRe + j \cdot \o\Rx\Tr \v\subIm \, ,
\label{eq:ChanCoeffProjDetail}
\end{align}
which exhibits a statistical dependence between the real and imaginary part because the random $\o\Rx$ affects both. In the following, we specify the statistics of $h$ in terms of the conditional PDF
$f(h \, | \, \o\Tx) = f(h \, | \, \v)$.

\begin{prop}\label{prop:PdfJoint}
Consider a random unit vector $\o\Rx \sim \UD(\mathcal{S})$, i.e. with uniform distribution on the unit sphere in $\mathbb{R}^3$, and a non-random vector $\v = \v\subRe + j \cdot \v\subIm \in \mathbb{C}^3$ with linearly independent $\v\subRe, \v\subIm \in \mathbb{R}^3$.
Let $v\subRe = \|\v\subRe\|$, $v\subIm = \|\v\subIm\|$, and
$\rho = \f{\v\subRe\Tr \v\subIm}{v\subRe v\subIm}$ (correlation coefficient).
Then the joint PDF of the projections
$\Re(h) = \o\Rx\Tr \v\subRe$
and
$\Im(h) = \o\Rx\Tr \v\subIm$
is
\begin{align}
& f(h \, | \, \v ) = f\big(\Re(h),\Im(h) \, | \, \v\subRe,\v\subIm \big) = 
\label{eq:PdfJoint} \\[2mm]
& \nonumber
\f{1}{v\subRe v\subIm \sqrt{1 \!-\! \rho^2}} \,
\CircFunc\Bigg(\, \bigg\| \!
\mtx{cc}{\! 1 \! & \,\,\,\, 0 \\ \! \rho \! & \!\! \sqrt{1 \!-\! \rho^2} \!\!}^{-1}\! 
\bigg[ \begin{array}{l} \! \Re(h) / v\subRe \!\! \\ \! \Im(h) / v\subIm \!\! \end{array} \bigg]
\bigg\|^2 \,\Bigg)
\end{align}
with $\CircFunc(x) = \f{1}{2\pi\sqrt{1 - x}} \IndicatorFunc_{[0,1]}(x)$. The uniform marginal distributions
$h\subRe \sim \UD(-v\subRe,v\subRe)$ and $h\subIm \sim \UD(-v\subIm,v\subIm)$ apply.
\end{prop}

\renewcommand\myCustomVec{\Big[ \begin{array}{l}a_o \\[-.8mm] b_o \end{array} \Big]}
\begin{proof}
We apply the Gram-Schmidt process to $\v\subRe, \v\subIm$ to obtain orthonormal vectors
\mbox{${\bf m} = \f{\v\subRe}{v\subRe}$},
${\bf n} = \f{(\eye_3 - {\bf mm}\Tr)\v\subIm}{\|(\eye_3 - {\bf mm}\Tr)\v\subIm\|}$.
They fulfill
$\v\subRe = v\subRe{\bf m}$
and
$\v\subIm = v\subIm \rho\,{\bf m} + v\subIm \sqrt{1-\rho^2}\,{\bf n}$.
Written as linear map,
$[\v\subRe\ \v\subIm] = [{\bf m\ n}] {\bf E}\Tr$.
The projections of $\o\Rx$ thus fulfill $[\Re(h)\ \Im(h)] = \o\Rx\Tr [{\bf m\ n}] {\bf E}\Tr = [m_o\ n_o] {\bf E}\Tr$.
The joint PDF $f_{m_o,n_o}$ is given by \Cref{prop:PdfJointOrth} below.
We subsequently obtain
the PDF of $\Re(h), \Im(h)$ with a change-of-variables argument:
for random $m_o, n_o$ with PDF $f_{m_o,n_o}$ and an invertible linear map ${\bf E}$, the PDF
$f(\Re(h), \Im(h) \, | \, \v)
= \f{1}{\det({\bf E})} f_{m_o,n_o}(\,{\bf E}^{-1} [\Re(h)\ \Im(h)]\Tr)$ applies.
\end{proof}

\begin{lemma}\label{prop:PdfJointOrth}
Consider orthonormal vectors ${\bf m}, {\bf n} \in \mathbb{R}^3$ and a random unit vector ${\bf o} \sim \UD(\mathcal{S})$. The joint PDF of $m_o = \o\Tr {\bf m}$, $n_o = \o\Tr {\bf n}$ is then given by
$f_{m_o,m_o}(m_o, n_o)
= \CircFunc( m_o^2 + n_o^2 )$.
\end{lemma}

\ifdefined\DoSixPageSubmission
For the proof of \Cref{prop:PdfJointOrth} we refer to \cite[Lemma~4.5]{Dumphart2020}.
\else
\begin{proof}
Due to the symmetry of $\mathcal{S}$, we can prove the statement by deriving the joint PDF of $o_1 = [1\,0\,0] \o$ and $o_2 = [0\,1\,0] \o$. We use polar coordinates $(o_1,o_2) = (R\cos\phi,R\sin\phi)$ and note that $R = \sqrt{o_1^2 + o_2^2} = \sqrt{1 - o_3^2}$. The marginal PDF $f(o_3) = \f{1}{2}$ leads to $f_R(R) = \f{R}{\sqrt{1-R^2}}$ for the radius $R \in [0,1]$. On the other hand, conditioned on $o_3$, the pair $(o_1,o_2)$ has uniform distribution on a circle of radius $R$; hence $f_{\phi|R} = f_{\phi} = \f{1}{2\pi}$, yielding the joint PDF $f_{R,\phi} = f_R \cdot f_{\phi|R} = \f{R}{2\pi\sqrt{1-R^2}}$.
Now $f_{o_1,o_2}$ follows from a change of variables from $r,\phi$ to $o_1,o_2$, i.e. a bijective map $[0,1] \times [0,2\pi] \rightarrow [-1,1]^2$. We obtain $f_{o_1,o_2} = \f{1}{R} f_{R,\phi} = \f{1}{2\pi\sqrt{1-R^2}}$ with $R^2 = o_1^2 + o_2^2$ via the Jacobian determinant.
%
\end{proof}
\fi
We note that the distribution $h \, | \, \o\Tx$ is equivalent to $h \, | \, \DoD\Tr \o\Tx$ because $\o\Rx \sim \UD(\mathcal{S})$ has rotational invariance. Thus,
\begin{align}
f(h \, | \, \o\Tx)
= f(h \, | \, \v)
= f(h \, | \, \DoD\Tr \o\Tx)
\, .
\label{eq:PdfEquivDetail}
\end{align}
An evaluation of this conditional PDF is shown in \Cref{fig:hTransConditional} for the exemplary value $\DoD\Tr \o\Tx = 0.3$.


\begin{prop}\label{prop:TransitionRegionCDF}
Let $h$ be distributed according to \Cref{prop:PdfJoint}. Then the CDF $F_{|h|^2}$ is within the bounds 
\begin{align}
\f{s}{2b} \,\leq\, 
F_{|h|^2|\v}(s | \v) \,\leq\,
\f{s}{2b} \, \left(1 - \f{s}{s_0}\right)^{-1/2}
\label{eq:hFadeJoint}
\end{align}
if $s < s_0$, whereby $s_0 = a -  \sqrt{a^2 - b^2}$ with
$a = \f{1}{2}(v\subRe^2 + v\subIm^2)$ and
$b = v\subRe v\subIm \sqrt{1-\rho^2}$.
\end{prop}

\begin{proof}
The joint PDF of $[\Re(h)\ \Im(h)] = \o\Rx\Tr [\v\subRe\ \v\subIm]$ is given by \Cref{prop:PdfJoint}. There, a linear map ${\bf E} \in \mathbb{R}^{2 \times 2}$ maps from the closed unit disk to the ellipse that is the support of $f(h | \v)$.
Let $s_0$ be the smaller eigenvalue of ${\bf E}\Tr {\bf E}$; the stated formula is obtained from the characteristic polynomial. Now $s < s_0$ guarantees that $f(h | \v) < \infty$ because then $h$ is in the interior of $\mathrm{supp}\,f(h | \v)$.
In particular,
$f(h | \v) = \f{\CircFunc(\|{\bf E}^{-1} [\Re(h)\ \Im(h)]\Tr\|^2)}{b} \leq \f{\CircFunc(s/s_0)}{b}$.
We find the upper bound via
$\Prob{|h|^2 \leq s | \v}
= \int_{|h|^2 \leq s} f(h|\v) dh
\leq \f{\CircFunc(s/s_0)}{b} \int_{|h|^2 \leq s} dh
= \f{\CircFunc(s/s_0) \cdot \pi s}{b}
= \f{s}{2b} /  \sqrt{1 - s/s_0}$.
Analogously, the lower bound is due to $f(h|\v) \geq \f{\psi(0)}{b} = \f{1}{2\pi b}$ for $s < s_0$.
\end{proof}

In essence, \Cref{prop:TransitionRegionCDF} states that $F_{|h|^2}(s) \propto s$ for small $s$ in the transition region. In contrary, the near- and far-field region behavior $F_{|h|^2}(s) \propto \sqrt{s}$ from \Cref{prop:MarginalRegionCDF} exhibits a larger concentration of probability mass near $|h|^2 = 0$.
The advantage of the transition region is caused by the sum of phase-shifted field vectors in \Cref{eq:v} providing polarization diversity: a deep fade $h=0$ can only occur if $\o\Rx$ is orthogonal to both $\v\subRe$ and $\v\subIm$. In other words, the field vector $\Re(\v e^{j2\pi ft})$ now oscillates on an ellipse, not on a line \cite[Sec.~2.5]{Orfanidis2002}.

\begin{figure}[!ht]\centering
\vspace{-2.5mm}
\includegraphics[height=32mm]{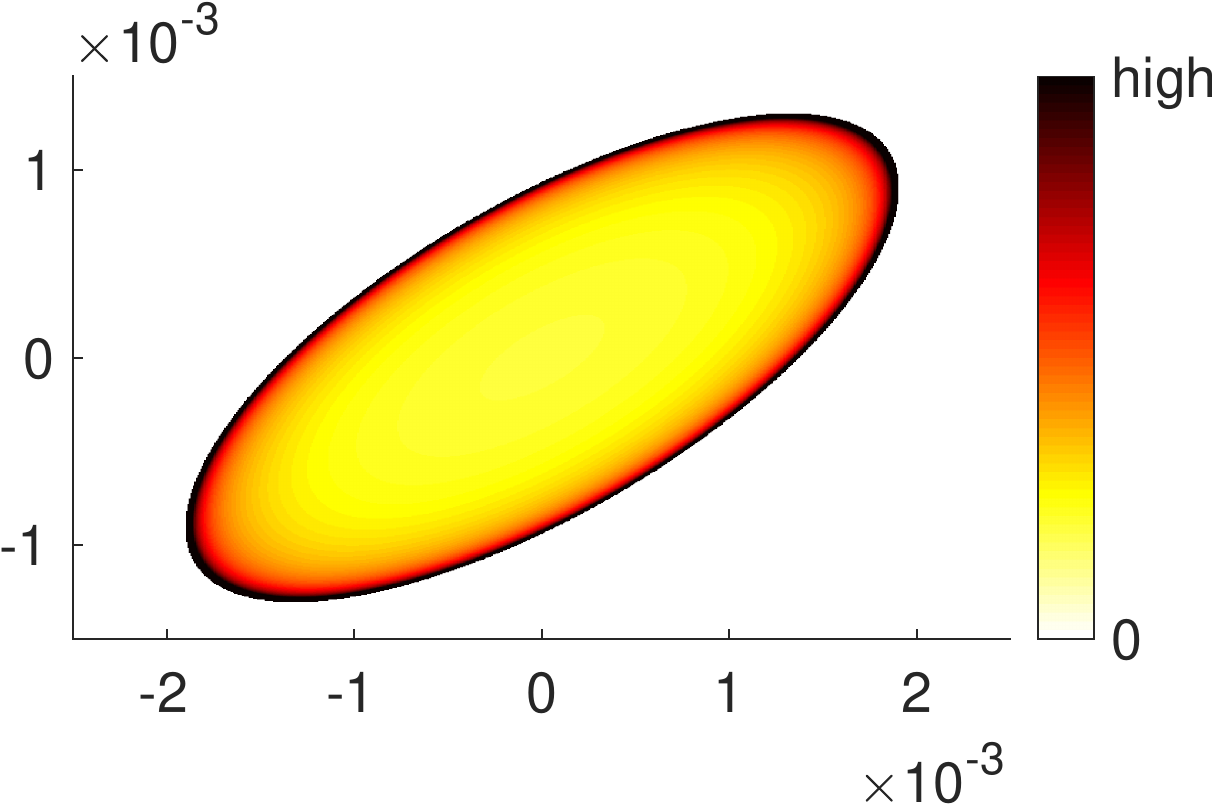}
\put(6,22){\rotatebox{90}{\small{conditional PDF}}}
\put(18,22){\rotatebox{90}{\scriptsize{$f(h\,|\,\DoD\Tr\o\Tx = 0.3)$}}}
\put(-151,44){\rotatebox{90}{\footnotesize{$\Im(h)$}}}
\put(-86,0){\footnotesize{$\Re(h)$}}
\vspace{-1.5mm}
\caption{Conditional PDF of the channel coefficient $h \in \mathbb{C}$ for random RX orientation $\o\Rx$, described by \Cref{prop:PdfJoint} in closed form. This evaluation assumes the values $\DoD\Tr\o\Tx = 0.3$, $kr = 2$, $\bar\alpha = 10^{-2}$.}
\label{fig:hTransConditional}
\end{figure}
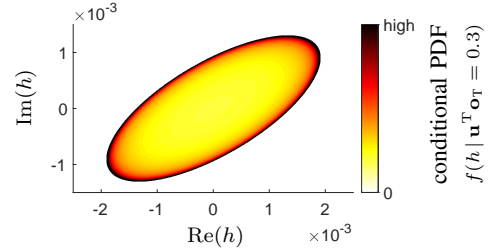


\subsection{Near-Far Transition, Random Orientations at Both Ends}

Finally, we address the distribution of $h$ in the fully-random case $\o\Tx, \o\Rx \iid \UD(\mathcal{S})$, again under consideration of all terms in \eqref{eq:hDipole}.
\Cref{fig:hFullyRandOrScatter} shows scatter plots of $h$ for various $kr$ values.

\begin{figure}[!h]\centering
\subfloat[$kr = 0.5$]{\label{fig:hFullyRandOrScatter_a}
\includegraphics[height=12.4mm]{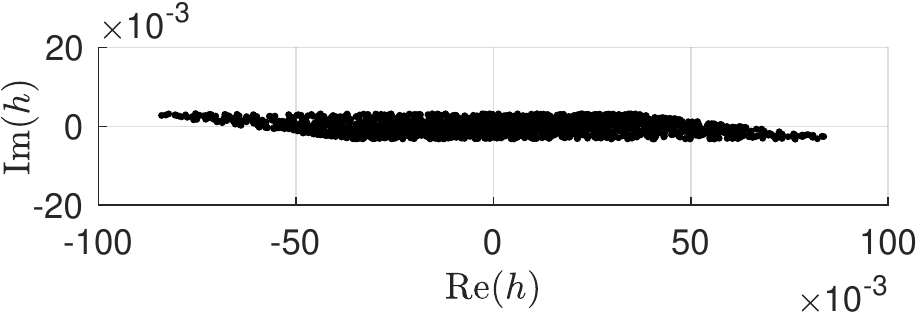}}\,\,\,\,\,\,%
\subfloat[$kr = 1$]{\label{fig:hFullyRandOrScatter_b}
\includegraphics[height=12.4mm]{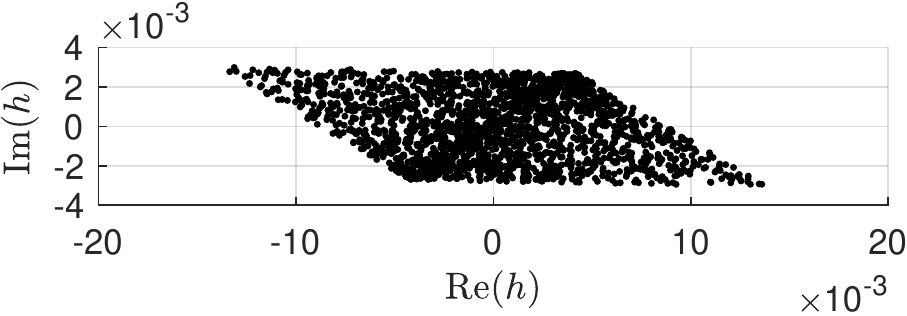}}
\\
\subfloat[$kr = 2$]{\label{fig:hFullyRandOrScatter_c}
\includegraphics[height=24.8mm]{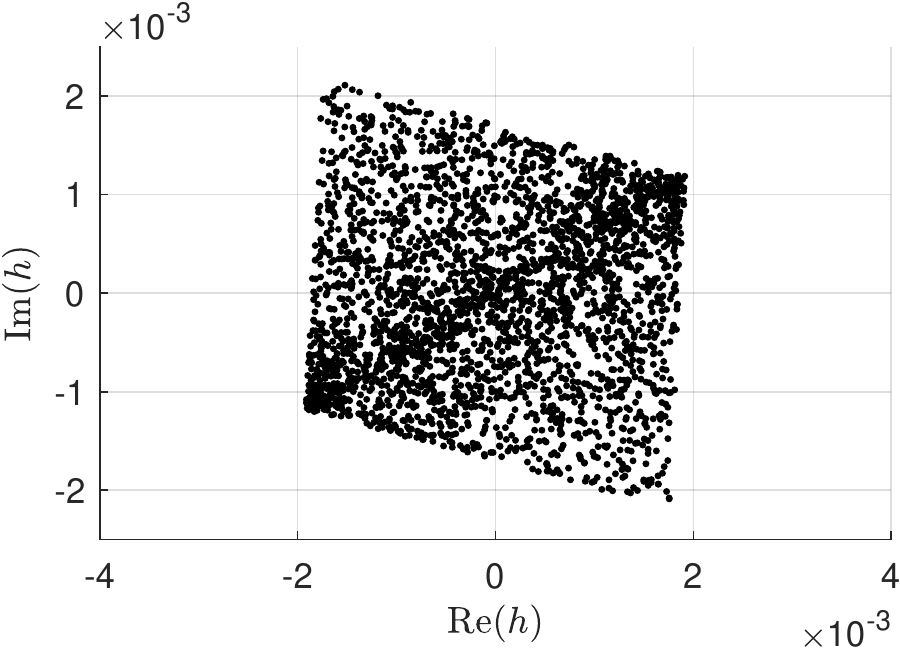}}\,\%
\subfloat[$kr = 5\pi$]{\label{fig:hFullyRandOrScatter_d}
\includegraphics[height=24.8mm]{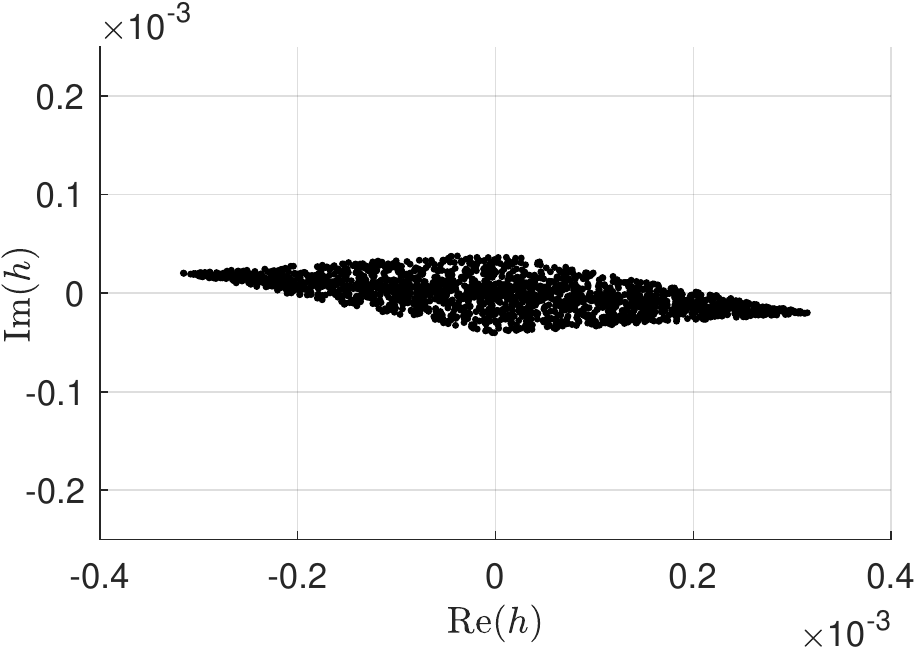}}
\caption{Scatter plots of the random channel coefficient $h \in \mathbb{C}$ between two dipoles with random orientations $\o\Tx, \o\Rx \sim \UD(\mathcal{S})$ for different regions. For $kr \ll kr_\tn{th}$ or $kr \gg kr_\tn{th}$, all samples lie on a line. The plots were obtained with random sampling and $\bar\alpha = 10^{-2}$ was assumed.}
\label{fig:hFullyRandOrScatter}
\end{figure}

\begin{prop}
Under \Cref{eq:UniformOrient} and based on the conditional PDF $f(h | \DoD\Tr \o\Tx) = f(h | \v)$ from \Cref{prop:PdfJoint}, the PDF of $h$ is given by
\begin{align}
f(h) = \f{1}{2} \cdot \int_{-1}^{+1} f(h \,|\, \DoD\Tr \o\Tx = x) \, dx \, .
\label{eq:NumIntegral}
\end{align}
\end{prop}

\begin{proof}
We find
$f(h) = \int_{-1}^{+1} f(h | \DoD\Tr \o\Tx = x) f(\DoD\Tr \o\Tx = x) dx$
from \Cref{eq:PdfEquivDetail} and marginalization.
With $\o\Tx \sim \UD(\mathcal{S})$ while $\DoD$ is non-random, the uniform distribution $\DoD\Tr \o\Tx \sim \UD(-1,1)$ applies (see the proof of \Cref{prop:JDistr}). Thus $f(\DoD\Tr \o\Tx) = \f{1}{2}$.
\end{proof}


A closed-form solution of the integral \Cref{eq:NumIntegral} is unavailable, but an evaluation obtained with numerical integration is shown in \Cref{fig:FinalPdfNumInt}. The results are supplemented by the geometric explanation of the rhombus-shaped support of the distribution in \Cref{fig:FinalPdfSupport}.

\begin{figure}[!ht]\centering\subfloat[fully-random-case PDF $f(h)$ for\\{\color{white}......}$kr=2$ (near-far-field transition)]{
\label{fig:FinalPdfNumInt}
\includegraphics[height=37mm]{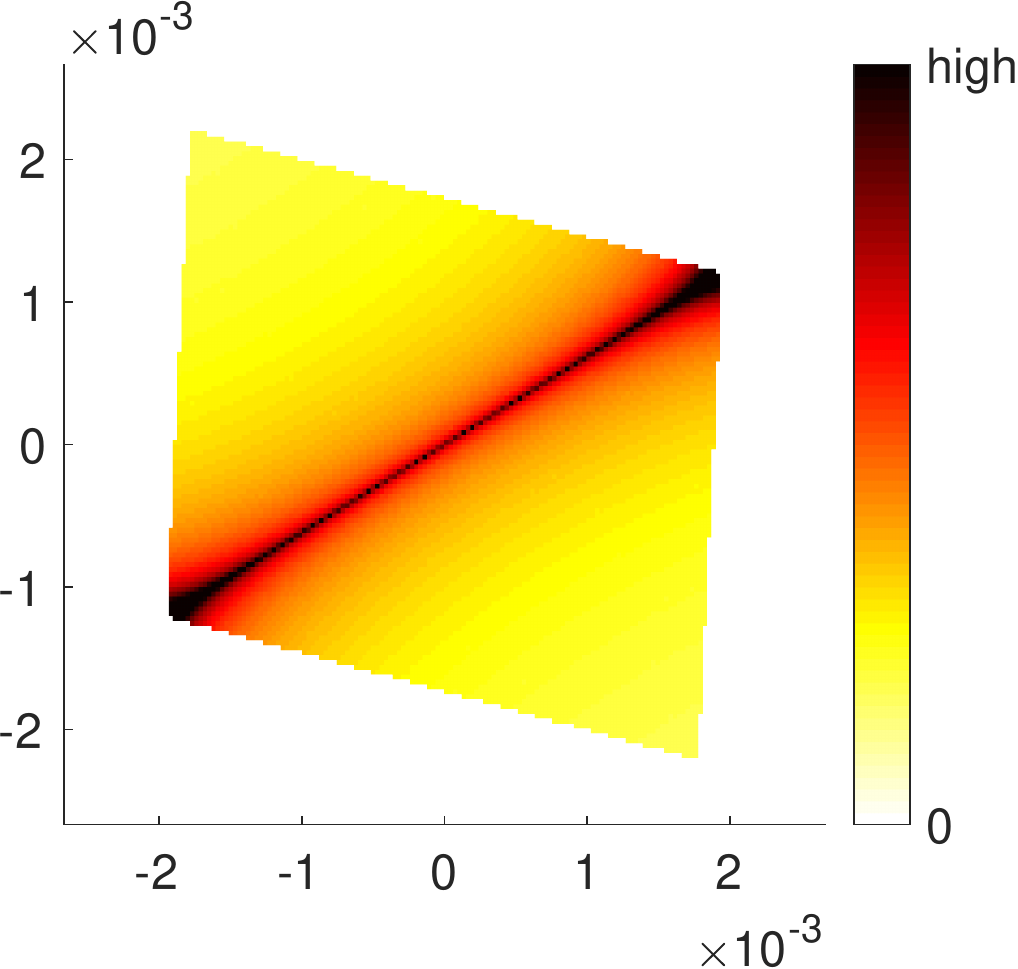}
\put(-3,18){\rotatebox{90}{\scriptsize{channel coeff. PDF $f(h)$}}}
\put(-70,-1){\scriptsize{$\Re(h)$}}
\put(-119.5,49){\rotatebox{90}{\scriptsize{$\Im(h)$}}}
}
\ \ \ \ \
\subfloat[rhombus-shaped support of\\{\color{white}.....}the PDF $f(h)$ for $kr=2$]{
\label{fig:FinalPdfSupport}
\includegraphics[height=37mm]{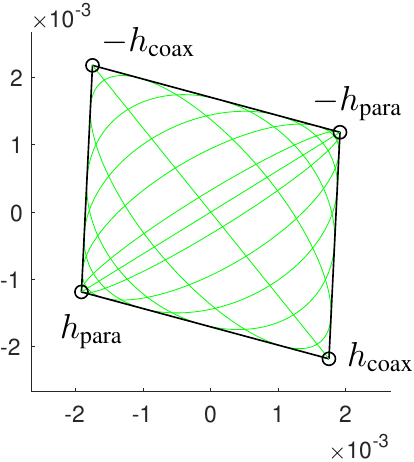}
\put(-53,-1){\scriptsize{$\Re(h)$}}
\put(-102.5,49){\rotatebox{90}{\scriptsize{$\Im(h)$}}}
}
\caption{PDF of the random channel coefficient $h \in \mathbb{C}$ for random dipole orientations $\o\Tx, \o\Rx \sim \UD(\mathcal{S})$, evaluated here for $kr = 2$ and $\bar\alpha = 10^{-2}$. The PDF was computed by solving \Cref{eq:NumIntegral} numerically.
\Cref{fig:FinalPdfSupport} shows how the rhombus-shaped support $\mathrm{supp}\,f(h)$ arises from a union of ellipses $\mathrm{supp}\,f(h \, | \, \DoD\Tr \o\Tx)$, shown for $\DoD\Tr \o\Tx \in \{0,\, .1,\, .3,\, .5,\, .7,\, .9,\, 1\}$. For $\DoD\Tr \o\Tx \in \{0,\, 1\}$ the ellipse becomes a line. The ellipses for $\DoD\Tr \o\Tx \approx 0$ cause a concentration of probability mass between $\pm h\Copl$.}
\label{fig:hSomeSupportIllu}
\end{figure}

\begin{figure}[!ht]\centering
\includegraphics[width=.96\columnwidth,trim=0 5 0 0]{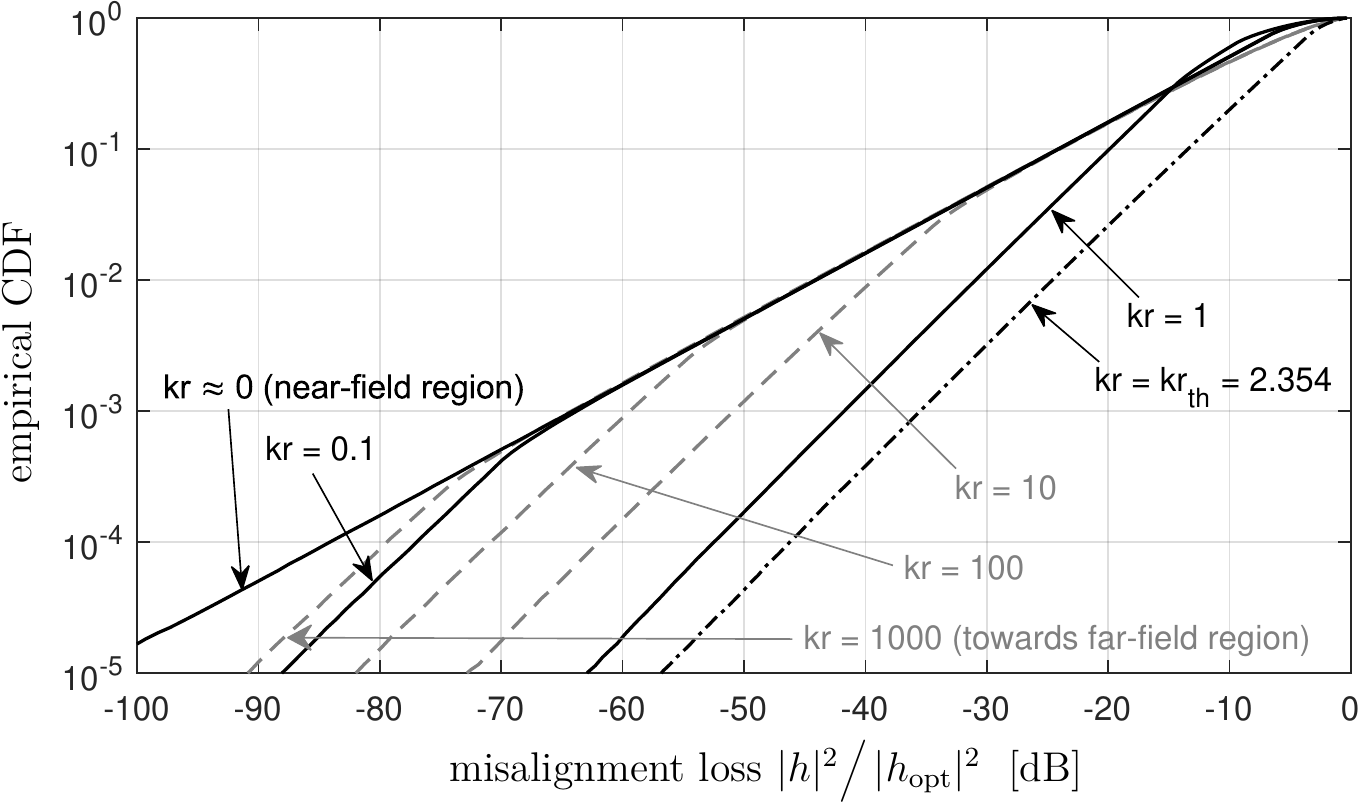}
\caption{Statistics of the channel attenuation due to random TX and RX orientations, computed via Monte Carlo simulation for i.i.d. uniform distributions in 3D. We observe that severe misalignment loss occurs with significant probability, especially in the near-field region and the far-field region. 
The near-far-field transition features a beneficial polarization diversity effect.}
\label{fig:StatMotivation}
\end{figure}

We argue that, via \Cref{eq:NumIntegral}, the beneficial property
$F_{|h|^2}(s) \propto s$
for small $s$ carries over from
$F_{|h|^2 | \v}(s\,|\,\v) \propto s$
in \Cref{prop:TransitionRegionCDF}. This is supported by the Monte-Carlo simulation in \Cref{fig:StatMotivation}, but a rigorous argument is unavailable.
A non-rigorous argument is that polarization diversity (i.e. linearly independent $\v\subRe, \v\subIm$) occurs with probability $1$. Put differently, a problematic case $\DoD\Tr \o\Tx = 0$ or $\DoD\Tr \o\Tx = \pm 1$, where the boundary ellipse of $\mathrm{supp}\, f_{h|\DoD\Tr \o\Tx}$ degenerates to a line, occurs with probability $0$.

%% file: 04-Outage.tex
\subsection{Outage Power Transfer Efficiency $\eta_\OutProb$}

The power transfer efficiency (PTE) $\eta = |h|^2$ is a random variable in the context of this paper. Analogous to the concept of outage capacity \cite{Tse2005}, we consider the outage PTE, defined as the PTE value $\eta_\OutProb$ for which an outage event $|h|^2 < \eta_\OutProb$ occurs with a certain probability $\OutProb$. The CDF of $|h|^2$ describes this very dependence: $\OutProb = F_{|h|^2}(\eta_\OutProb)$.

\begin{prop}
\label{prop:OutagePTEStatements}
Assume \Cref{eq:UniformOrient} and that the RX is in the near-field region ($J_* = J\nf$) or the far-field region ($J_* = J\ff$). Then, a target PTE $\eta_\OutProb$ results in an outage probability
\begin{align}
\OutProb = F_{|h|^2}(\eta_\OutProb)
\approx 
2\cdot f_{J_*}\!(0) \sqrt{\f{\eta_\OutProb}{\eta\Opt}} \, .
\label{eq:OutProbPTE}
\end{align}
Vice versa, a target outage probability $\OutProb$ yields the outage PTE
\begin{align}
\eta_\OutProb = F_{|h|^2}^{-1}( \OutProb )
\approx
\f{\OutProb^2 \, \eta\Opt}{(2\cdot f_{J_*}\!(0))^2} \, .
\label{eq:OutagePTESpecific}
\end{align}
The approximations are accurate for $\eta_\OutProb \ll \eta\Opt$.
\end{prop}

\begin{proof}
The statements follow directly from \Cref{prop:MarginalRegionCDF}.
\end{proof}

From \Cref{eq:OutProbPTE,eq:OutagePTESpecific} we observe the proportionality
$\OutProb \propto \eta\Opt^{-\f{1}{2}}$
as well as
$\eta_\OutProb \propto \OutProb^2$.
This scaling behavior demonstrates the drastic fading effect due to random antenna orientations. 
On the one hand, increasing $\eta\Opt$ (e.g., by improving technical parameters or reducing the distance) is not an efficient means for reducing $\OutProb$.
On the other hand, requiring some degree of reliability (i.e. a small $\OutProb$) is associated with an extremely small PTE $\eta_\OutProb$. For example, aiming for a factor-$10$ improvement of $\OutProb$ demands a $20\dB$ loss for $\eta_\OutProb$.


The situation improves in the near-far-field transition:
there, $F_{|h|^2}(s) \propto s$ holds for small $s$ (cf. \Cref{prop:TransitionRegionCDF}), which results in the more beneficial proportionalities
$\OutProb \propto \eta\Opt^{-1}$
and
$\eta_\OutProb \propto \OutProb$. The improvement stems from polarization diversity.

\subsection{Outage Capacity $C_\OutProb$}

We shift our focus to narrowband data communication over this fading channel, with transmit power $P\Tx$ and reception in additive white Gaussian noise (AWGN) of power $P\N$. The signal-to-noise ratio $\SNR = |h|^2 P\Tx / P\N$ is random and the instantaneous channel capacity
$C = \log_2(1 + \SNR)$,
measured in $\mathrm{bit/s/Hz}$,
is thus also random and can fade to zero.

A well-established measure for the communication performance of a fading channel is the outage capacity \cite[Eq.~5.57]{Tse2005}
\begin{align}
C_\OutProb = \log_2\!\bigg( 1 + \f{F_{|h|^2}^{-1}( \OutProb ) \cdot P\Tx}{P\N} \bigg)
\label{eq:OutageCapacity}
\end{align}
for which, by definition, the event $\log_2(1+\SNR) < C_\OutProb$ occurs with probability $\OutProb$. We argue that $C_\OutProb \propto F_{|h|^2}^{-1}( \OutProb )$ for small $\OutProb$ because the bound
$C_\OutProb \leq \log_2(e) \cdot F_{|h|^2}^{-1}( \OutProb ) \cdot P\Tx/P\N$,
which is obtained through log-linearization,
is tight for low SNR or for a small target $\OutProb$. Hence, by \Cref{eq:OutagePTESpecific}, the near- and the far-field regions exhibit the scaling behavior
\begin{align}
C_\OutProb \propto \OutProb^2 .
\end{align}
This means that a target outage probability
$\OutProb \ll 1$
can only be achieved with an extremely small data rate.
In contrary, the near-far-field transition exhibits the more beneficial scaling behavior $C_\OutProb \propto \OutProb$ due to polarization diversity (by \Cref{prop:TransitionRegionCDF}).

\subsection{Bit Error Rate $\BER$}

Another popular measure of the communication performance over a fading channel is the bit error rate $\BER$. For antipodal modulation (BPSK) and reception in AWGN, its value given $h$ is $Q(\sqrt{2\,\SNR})$ whereby $\SNR = |h|^2 P\Tx / P\N$ is random and subject to fading. \cite[Eq.~3.13]{Tse2005}

\begin{prop}\label{prop:BERUpperBound}
Assume \Cref{eq:UniformOrient}, AWGN, and either the near-field region ($J_* = J\nf$) or the far-field region ($J_* = J\ff$). Then, the bit error rate of BPSK modulation has the upper bound
\begin{align}
\BER < \f{f_{J_*}\!(0)}{\sqrt{\pi\cdot\SNR\Opt}}
\label{eq:BERUpperBound}
\end{align}
which becomes tight for large $\SNR\Opt = |h\Opt|^2 P\Tx / P\N$.
\end{prop}

\begin{proof}
$\SNR \approx J_*^2\,\SNR\Opt$
applies in the near- or far-field region. For the bit error rate we calculate $\BER = \EV{ Q\Big(\sqrt{2J_*^2 \,\SNR\Opt} \,\Big) }
= \int_{-1}^1 f_{J_*}\!(J_*) \, Q(\!\sqrt{2J_*^2 \SNR\Opt} ) dJ_*
\leq 2\cdot f_{J_*}\!(0) \int_0^1  \, Q(\!\sqrt{2J_*^2 \SNR\Opt} ) dJ_*
= 2\cdot f_{J_*}\!(0) \big( \f{1 - e^{-\SNR\Opt}}{\sqrt{4\pi\cdot\SNR\Opt}} + Q(\sqrt{2\SNR\Opt}\,) \big)$%
\ifdefined\DoSixPageSubmission
 and apply $Q(x) < \f{1}{x} \f{1}{\sqrt{2\pi}} e^{-x^2/2}$.
\else
. Applying the well-known $Q$-function bound $Q(x) < \f{1}{x} \f{1}{\sqrt{2\pi}} e^{-x^2/2}$ concludes the proof. 
\fi
%
%
\end{proof}

The large-SNR description in \Cref{prop:BERUpperBound} has the standard form $\BER \propto \SNR\Opt^{-L}$ from \cite[Eq.~3.158]{Tse2005}. We deduce that the diversity exponent is $L = \f{1}{2}$ for the near- and far-field regions, associated with catastrophic fading (worse than $L = 1$ of Rayleigh fading). In a similar fashion, it can be argued that $\BER \propto \SNR\Opt^{-1}$  in the near-far-field transition for large $\SNR\Opt$, i.e. $L = 1$ (like Rayleigh fading). This is a direct consequence of $f_{|h|}(x) \propto x$ for small $x$; the details are omitted.


%


%% file: 99-Appendix.tex
The channel coefficient $h$ in \Cref{eq:hDipole} relates the power wave emitted by a power-matched TX amplifier to the power wave into a power-matched RX amplifier or tank circuit. A necessary condition for \Cref{eq:hDipole} to apply is that $h$ (obtained with this formula) fulfills $|h|^2 \ll 1$,
\ifdefined\DoSixPageSubmission
i.e. the dipoles are weakly coupled.
\else
i.e. the dipoles are only weakly coupled and no simultaneous matching problem occurs \cite{RaholaTCS2008}.
We note that $|h|^2 \ll 1$ is fulfilled naturally if $h$ is in a deep fade, hence many aspects of the outage analysis in \Cref{sec:outage} also apply to short-distance links.
A detailed exposition is given in \cite{Dumphart2020}.
\fi

The prefactor $\coeffH = \bar\coeffH\,e^{-jkr}$ in \Cref{eq:hDipole} comprises technical link parameters in $\bar\coeffH$. For \textit{loop antennas} (i.e. coils), the value
$\bar\coeffH = \f{j \mu_0 A\Tx \NTurnTx A\Rx \NTurnRx f k^3}{\sqrt{4 R\Tx R\Rx}}$
applies if the coils are electrically small, the turn pitch angle is small, and the coil diameters are significantly smaller than $r$. We use the permeability $\mu_0$, TX- and RX-side number of turns $\NTurnTx$ and $\NTurnRx$, the coil areas $A\Tx$ and $A\Rx$, carrier frequency $f$, wavenumber $k$, and the TX- and RX-side antenna resistances $R\Tx$ and $R\Rx$. \cite[Sec.~5.2]{Balanis2005}

We note that
$\f{\mu_0}{2\pi} A\Tx \NTurnTx A\Rx \NTurnRx \, r^{-3} J\nf$ is the mutual inductance $M$. Furthermore, a small power-matched TX coil is described by a magnetic dipole moment phasor ${\bf m} = A\Tx \NTurnTx i\Tx \o\Tx$ with current phasor $i\Tx = \sqrt{P\Tx / R\Tx}$ and TX power $P\Tx$.

Between \textit{dipole antennas} without ohmic losses, $\bar\coeffH$ is given by the antenna directivity: $\bar\coeffH = 1.5$ for electrically small dipole antennas and $\bar\coeffH \approx 1.64$ for the $\lambda/2$-length case \cite[Cpt.~4]{Balanis2005}.
\ifdefined\DoSixPageSubmission
\else
The formalism of this paper applies to $\lambda/2$-dipole antennas because their precise behavior (see \cite{Balanis2005})
$\beta\ff = \cos(\f{\pi}{2} \cos(\theta\Tx)) / \sin(\theta\Tx)$
is approximated by \Cref{eq:bFfMagn} which can be written as $\beta\ff = \sin(\theta\Tx)$.
\fi
